\documentclass[10pt,twocolumn]{article}

\usepackage[utf8]{inputenc}
\usepackage{amsmath, amssymb, amsfonts, amsthm} 
\usepackage{titlesec}
\usepackage{cite}
\usepackage{booktabs}
\usepackage{graphicx}
\usepackage[colorlinks=false]{hyperref}
\usepackage[format=plain,labelfont=it]{caption}
\usepackage[left=1.5cm,right=1.5cm,top=2cm,bottom=2cm]{geometry}
\usepackage{color}
\usepackage{fancyhdr}
\usepackage{xurl}
\usepackage{svg}
\usepackage{nicefrac}
\usepackage{siunitx}
\usepackage{mathrsfs}
\usepackage{setspace}

\titleformat{\section}{\centering\normalfont\scshape}{\Roman{section}.}{5pt}{}
\titleformat{\subsection}{\normalfont\it}{\Alph{subsection}.}{5pt}{}
\titleformat{\subsubsection}{\normalfont\it}{\hspace{4mm}\arabic{subsubsection})}{5pt}{}

\newcommand\infoFootnote[1]{%
  \begingroup
  \renewcommand\thefootnote{}\footnote{#1}%
  \addtocounter{footnote}{-1}%
  \endgroup}

\newtheorem{thm}{Theorem}

\newtheorem{lem}[thm]{Lemma}
\newtheorem{prop}[thm]{Proposition}
\newtheorem{assum}{Assumption}

\newtheorem{defn}{Definition}

\newtheorem{rem}{Remark}

\DeclareMathOperator*{\argmin}{\arg\min}

\newcommand{\R}{\mathbb{R}} 

\newcommand{\ab}{\boldsymbol{a}}
\newcommand{\gb}{\boldsymbol{g}}
\newcommand{\pb}{\boldsymbol{p}}
\newcommand{\ub}{\boldsymbol{u}}
\newcommand{\vb}{\boldsymbol{v}}
\newcommand{\wb}{\boldsymbol{w}}
\newcommand{\xb}{\boldsymbol{x}}
\newcommand{\zb}{\boldsymbol{z}}
\newcommand{\zerob}{\boldsymbol{0}}
\newcommand{\oneb}{\boldsymbol{1}}

\newcommand{\Ab}{\boldsymbol{A}}
\newcommand{\Bb}{\boldsymbol{B}}

\newcommand{\Gb}{\boldsymbol{G}}
\newcommand{\Ib}{\boldsymbol{I}}
\newcommand{\Kb}{\boldsymbol{K}}
\newcommand{\Lb}{\boldsymbol{L}}
\newcommand{\Mb}{\boldsymbol{M}}
\newcommand{\Nb}{\boldsymbol{N}}
\newcommand{\Pb}{\boldsymbol{P}}
\newcommand{\Qb}{\boldsymbol{Q}}
\newcommand{\Rb}{\boldsymbol{R}}
\newcommand{\Sb}{\boldsymbol{S}}
\newcommand{\Ub}{\boldsymbol{U}}
\newcommand{\Vb}{\boldsymbol{V}}
\newcommand{\Xb}{\boldsymbol{X}}

\newcommand{\Pib}{\boldsymbol{\Pi}}

\newcommand{\Dbc}{\boldsymbol{\mathcal{D}}}

\newcommand{\Nc}{\mathcal{N}}
\newcommand{\Sc}{\mathcal{S}}
\newcommand{\Tc}{\mathscr{T}}
\newcommand{\Hc}{\mathcal{H}}

\newcommand{\halff}{{\nicefrac{1}{2}}}

\newcommand{\rank}[1]{\mathrm{rank}\left({#1}\right)}
\newcommand{\trace}[1]{\mathrm{tr}({#1})}
\newcommand{\froSq}[1]{\|{#1}\|_{\mathrm{F}}^2}
\newcommand{\fro}[1]{\|{#1}\|_{\mathrm{F}}}

\newcommand{\N}{\mathbb{N}}

\renewcommand{\boldsymbol}[1]{#1}
\renewcommand{\mathbf}[1]{\mathrm{#1}}

\newcommand{\Acl}{\Ab_{\mathrm{cl}}}
\newcommand{\tAcl}{\widetilde\Ab_{\mathrm{cl}}}
\newcommand{\tKb}{\widetilde\Kb}
\newcommand{\Als}{\Ab_{\mathrm{LS}}}
\newcommand{\Bls}{\Bb_{\mathrm{LS}}}
\newcommand{\Kls}{\Kb_{\mathrm{LS}}}

\def\tvdots{\vbox{\baselineskip=2pt \lineskiplimit=0pt \kern6pt \hbox{.}\hbox{.}\hbox{.}}}

\title{\vspace{-2mm}\bf On the Effect of Quadratic Regularization \\ in Direct Data-Driven LQR}
\author{Manuel Kl\"adtke, Feiran Zhao, Florian D\"orfler, Moritz Schulze Darup \vspace{2mm}}
\date{}

\begin{document}
\maketitle

\begin{abstract}

This paper proposes an explainability concept for direct data-driven linear quadratic regulation (LQR) with quadratic regularization. Our perspective follows the \textit{parametric effect of regularization}, an analysis approach that translates regularization costs from auxiliary variables to system quantities, enabling intuitive interpretations. The framework further enables the elimination of auxiliary variables, thereby  reducing computational complexity. We demonstrate the effectiveness of our approach and the identified effect of regularization via simulations.
\end{abstract}
\infoFootnote{M. Kl\"adtke and M. Schulze Darup are with the \href{https://rcs.mb.tu-dortmund.de/}{Control and~Cyberphysical Systems Group}, Faculty of Mechanical Engineering, TU Dortmund University, Germany. E-mails:  \href{mailto:manuel.klaedtke@tu-dortmund.de}{\{manuel.klaedtke, moritz.schulzedarup\}@tu-dortmund.de}. \vspace{0.5mm}}
\infoFootnote{F. Zhao and F. D\"orfler  are with the \href{https://ee.ethz.ch/}{Department of Information Technology and Electrical Engineering}, ETH Zürich, Switzerland. (e-mails: \{dorfler, zhaofe\}@control.ee.ethz.ch)}
\infoFootnote{This paper is a \textbf{preprint} of a contribution to the 23rd IFAC World Congress 2026. }

\section{Introduction}

The linear quadratic regulator (LQR) is the cornerstone of optimal control and a natural object of interests for recent advances in direct data-driven control methods. 
These methods compute controllers directly from measured trajectories. By proposing a data-based closed-loop system parameterization, the seminal works \cite{DePersis2020, DePersis2021} reformulate the LQR problem in a direct data-driven fashion without involving any system identification. Recently, \cite{Zhao2025_DeePO-LQR} have proposed a sample covariance parameterization of the LQR to make the formulation in \cite{DePersis2021} amenable to real-time computation. Both parameterizations have seen fruitful applications in different control problems \cite{mejari2025bias, liu2024learning}.

While performing well under persistently exciting inputs, these approaches are sensitive to process noise. To this end, quadratic regularizations have been introduced to enhance robust stability or certainty-equivalence with performance certificates \cite{Dörfler2022_roleOfRegLQR, Dörfler2023_certaintyEqLQR}. Moreover, simulation results have shown that, for the covariance parameterization, the regularization can improve both performance and stability significantly over the certainty-equivalence LQR \cite{ Zhao2025_regularization}. 

In this paper, we take a new perspective on the impact of quadratic regularizations for the LQR problem. Motivated by our recent work \cite{Klaedtke2025_AT} on explainable data-driven predictive control~(DPC) \cite{Coulson2019DeePC}, which develops analysis tools that enable exact and intuitive interpretations for the effects of regularization on the predicted trajectories, we introduce the concept of the \textit{parametric effect of regularization} (formalized in Def.~\ref{def:parametricEffect} below) to clarify how (quadratic) regularization influences direct data-driven LQR solutions. This concept translates regularization costs from auxiliary variables to interpretable system quantities, thereby enhancing transparency and enabling computational simplifications. Analytical results and numerical experiments reveal distinct roles of individual cost terms in shaping closed-loop dynamics and controller gains, contributing to the understanding and improvement of data-driven LQR.

The remainder of this paper is organized as follows. Section~\ref{sec:basics} reviews preliminaries on data-driven LQR. Section~\ref{sec:main} presents the main results on the concept of the parametric effect of several quadratic regularization. Section~\ref{sec:experiments} provides numerical experiments illustrating and validating the theoretical findings. We conclude our work in Section~\ref{sec:Conclusion}.

\section{Preliminaries}\label{sec:basics}

Consider a linear time-invariant (LTI) system 
\begin{equation}
    \xb(k+1)=\Ab\xb(k) + \Bb \ub(k) + \wb(k)
    \label{eq:LTI}
\end{equation}
with states $\xb\in\R^n$, inputs $\ub\in\R^m$, process noise $\wb\in\R^n$, and $(\Ab, \Bb)$  stabilizable. We refer to \eqref{eq:LTI} with $\wb(k) = \zerob$ for all $k \in \N$ as the noiseless case, in contrast to the noisy case, where $\wb(k)$ may be nonzero. For the noisy case, no probabilistic assumptions are imposed on $\wb(k)$. To prepare our analysis in Section~\ref{sec:main}, we briefly review model-based and data-driven LQR.

\subsection{Model-based LQR} \label{eq:modelLQR}

Consider a performance signal
$$
    \zb(k) = 
    \begin{pmatrix}
        \Qb^\halff & \zerob \\
        \zerob & \Rb^\halff
    \end{pmatrix}
    \begin{pmatrix}
        \xb(k) \\
        \ub(k)
    \end{pmatrix}
$$
for the LTI system \eqref{eq:LTI}, where $\Qb$ and $\Rb$ are positive definite weight matrices. One way to phrase the LQR problem is designing a linear state-feedback $u(k) = \Kb^\ast \xb(k)$, such that $\Ab+\Bb \Kb^\ast$ is Schur stable and the $\Hc_2$-norm $\|\Tc\|_2$ of the transfer function $\Tc: \wb \to \zb$ of the closed-loop system
$$
    \begin{pmatrix}
        \xb(k+1) \\
        \zb(k)
    \end{pmatrix} = 
    \left(
    \begin{array}{c|c}
       \Ab+\Bb\Kb^\ast  & \Ib \\ \hline
        \begin{pmatrix}
            \Qb^\halff \\
            \Rb^\halff \Kb
        \end{pmatrix} & \zerob
    \end{array}
    \right)
   \begin{pmatrix}
       \xb(k) \\
       \wb(k)
   \end{pmatrix}
$$
is minimized. When $\Ab+\Bb\Kb^\ast$ is Schur stable, the corresponding squared $\Hc_2$-norm can be computed \cite{chen1995_sampledData} as
$$
    \|\Tc\|_2^2 = \trace{\Qb\Pb^\ast} + \trace{\Kb^{\ast, \top} \Rb \Kb^\ast \Pb^\ast},
$$
where $\Pb^\ast$ is the controllability Gramian of the closed-loop system, which coincides with the unique solution to the Lyapunov equation
$$
    (\Ab+\Bb\Kb^\ast)\Pb^\ast(\Ab+\Bb\Kb^\ast)^\top - \Pb^\ast + \Ib = \zerob.
$$
The optimal gain $\Kb^\ast$ is also unique and the pair $(\Pb^\ast, \Kb^\ast)$ can be obtained by solving the optimization problem  
\begin{align}
    &\min_{\Pb
    \succeq \Ib, \Kb} \trace{\Qb\Pb} + \trace{\Kb^\top \Rb \Kb \Pb} \label{eq:lqrOptModelBased}\\
    &\text{s.t.} \quad (\Ab+\Bb\Kb)\Pb(\Ab+\Bb\Kb)^\top - \Pb + \Ib 
    \preceq \zerob, \nonumber
\end{align}
which admits a convex semi-definite program (SDP) parameterization \cite{DePersis2021}.

\subsection{Data-driven closed-loop parameterizations} \label{sec:ddLQR}

If the system parameters $(\Ab, \Bb)$ are not available, a data-driven representation of the closed-loop $\Ab + \Bb \Kb$ may be obtained from raw data \cite{DePersis2020}. Consider the data matrices 
    $\Xb_0 := \begin{pmatrix}
    \xb_0^{(1)} & \hdots & \xb_0^{(\ell)}
\end{pmatrix}  \in \R^{n \times \ell}$, 
$\Ub_0 := \begin{pmatrix}
    \ub^{(1)} & \hdots & \ub^{(\ell)} 
\end{pmatrix}  \in \R^{m \times \ell} $, and
$\Xb_1 := \begin{pmatrix}
    \xb^{(1)} & \hdots & \xb^{(\ell)} 
\end{pmatrix}  \in \R^{n \times \ell}$, 
where each triple $(\xb_0^{(i)}, \ub^{(i)}, \xb^{(i)})$ was generated by the noiseless LTI system \eqref{eq:LTI}, i.e., we have 
\begin{equation}
    \Xb_1 = \Ab \Xb_0 + \Bb \Ub_0, \label{eq:dataLTI}
\end{equation}
For easier notation, we introduce the stacked data matrices $\Dbc_0:= \begin{pmatrix}
    \Xb_0^\top & \Ub_0^\top
\end{pmatrix}^\top$ and $\Dbc:= \begin{pmatrix}
    \Dbc_0^\top & \Xb_1^\top
\end{pmatrix}^\top$. Further, we assume that the data is persistently exciting in the sense that the input-state data matrix $\Dbc_0$ has full row-rank 
\begin{equation}
    \rank{\Dbc_0} = n+m. \label{eq:GPE}
\end{equation} 
If \eqref{eq:GPE} is satisfied, there exists a matrix 
$\Gb\in\R^{\ell \times n}$ such that 
$$
    \begin{pmatrix}
        \Ib \\
        \Kb
    \end{pmatrix} = \begin{pmatrix}
        \Xb_0 \\
        \Ub_0
    \end{pmatrix} \Gb.
$$
Therefore, right-multiplying \eqref{eq:dataLTI} with $\Gb$ yields 
$$
    \Xb_1 \Gb = \Ab \Xb_0 \Gb + \Bb \Ub_0 \Gb = \Ab + \Bb \Kb.
$$%
Alternatively, \cite{Zhao2025_DeePO-LQR} introduce a data-driven covariance parameterization of the closed-loop. Similarly, if \eqref{eq:GPE} holds, one can find a matrix $\Vb\in \R^{n \times n}$ such that 
$$
    \begin{pmatrix}
        \Ib \\
        \Kb
    \end{pmatrix} = \begin{pmatrix}
        \frac{1}{\ell}\Xb_0 \Dbc_0^\top \\
        \frac{1}{\ell}\Ub_0 \Dbc_0^\top
    \end{pmatrix} \Vb,
$$
and multiply \eqref{eq:dataLTI} with $\ell^{-1}\Dbc_0^\top\Vb$, which yields
$$
    \frac{1}{\ell}\Xb_1 \Dbc_0^\top\Vb = \Ab \cdot \frac{1}{\ell}\Xb_0 \Dbc_0^\top\Vb + \Bb \cdot \frac{1}{\ell} \Ub_0 \Dbc_0^\top\Vb = \Ab + \Bb \Kb.
$$
These parameterizations may be used to formulate a data-driven version of LQR, which is the focus of this work.

\section{Quadratic regularization in data-driven LQR}\label{sec:main}

In this section, We first introduce the data-driven LQR formulation and explain why it fails in a noisy setting without regularization. Subsequently, we introduce a new analysis tool to explain the effect of regularizations in data-driven LQR. Furthermore, we show how these results may also be used to eliminate auxiliary variables and hence reduce computational complexity.

\subsection{Necessity of regularization in data-driven LQR}

The considerations in the previous section lead to the two data-driven representations
\begin{equation}
    \begin{pmatrix}
        \Ib \\
        \Kb\\
        \Acl
    \end{pmatrix} =
    \begin{pmatrix}
        \Xb_0 \\
        \Ub_0 \\        
        \Xb_1
    \end{pmatrix}\Gb, \label{eq:ddPredCL}
\end{equation}
and
\begin{equation}
    \begin{pmatrix}
        \Ib \\
        \Kb\\
        \Acl
    \end{pmatrix} =
    \begin{pmatrix}
        \frac{1}{\ell}\Xb_0 \Dbc_0^\top\\
        \frac{1}{\ell}\Ub_0 \Dbc_0^\top\\        
        \frac{1}{\ell}\Xb_1 \Dbc_0^\top
    \end{pmatrix}\Vb \label{eq:ddPredCLcovar}
\end{equation}
of the closed-loop system. Here, we introduced the variable $\Acl$ to highlight that, while we have $\Acl = \Ab+\Bb\Kb$ in the noiseless setting, this need not be the case in noisy settings. Using \eqref{eq:ddPredCL}, a data-driven version of \eqref{eq:lqrOptModelBased} can be stated as
\begin{align}
    &\min_{\Pb \succeq \Ib, \Kb, \Acl, \Gb} \trace{\Qb\Pb} + \trace{\Kb^\top \Rb \Kb \Pb} + H(\Gb, \Pb) \label{eq:lqrOptDD}\\
    &\qquad \text{s.t.} \;\; \Acl\Pb\Acl^\top - \Pb + \Ib \preceq \zerob, \nonumber 
     \quad \,(\Kb, \Acl, \Gb) \; \text{satisfy} \; \eqref{eq:ddPredCL}, \nonumber
\end{align}
and an analogous version follows for \eqref{eq:ddPredCLcovar}. We denote the optimizers of \eqref{eq:lqrOptDD} as $\overline\Kb, \overline\Ab_\text{cl}, \overline{\Pb}$ to distinguish them from the true model-based LQR solution $\Kb^\ast, \Pb^\ast$. In \eqref{eq:lqrOptDD}, $H(\Gb, \Pb)$ is a regularizer used to robustify the solution in the noisy data case. In the noiseless setting with data satisfying \eqref{eq:GPE}, the data matrices always satisfy 
$\rank{\begin{matrix}
    \Dbc
\end{matrix}} = \rank{\begin{matrix}
    \Dbc_0
\end{matrix}}, $
such that \eqref{eq:ddPredCL} implies a unique (and exact) mapping from $\Kb$ to $\Acl$. Hence, for $H(\Gb, \Pb) = \zerob$, the optimal pair $(\overline\Pb, \overline\Kb)$ obtained from \eqref{eq:lqrOptDD} coincides \cite{DePersis2021} with the true LQR solution $(\Pb^\ast, \Kb^\ast)$.  However, in more realistic, noisy settings, trying to obtain $(\Pb^\ast, \Kb^\ast)$ via \eqref{eq:lqrOptDD} with $H(\Gb, \Pb) = 0$ may fail spectacularly, which we specify in Proposition~\ref{prop:fail} below. In this noisy setting, the rank deficiency $\rank{\begin{matrix}
            \Dbc
        \end{matrix}} = \rank{\begin{matrix}
            \Dbc_0
        \end{matrix}}$ is usually lost, and instead there are $\rank{\begin{matrix}
            \Dbc
        \end{matrix}} - \rank{\begin{matrix}
            \Dbc_0
        \end{matrix}}> 0$ 
 degrees of freedom to choose $\Acl$ independently of $\Kb$, which are greedily exploited in \eqref{eq:lqrOptDD}. To reflect this crucial aspect in our analyses, we make the following assumption.
\begin{assum}\label{assum:fullRank}
    The data matrix $\Dbc$ has full row-rank.
\end{assum}
Note that Assumption~\ref{assum:fullRank} also requires \eqref{eq:GPE} to be satisfied. Further, it means that $\wb(k)$ has artificially excited all (even otherwise physically impossible) directions of the state-space. Crucially, results based on Assumption~\ref{assum:fullRank} extend even beyond noisy LTI systems \eqref{eq:LTI}. In fact, Assumption~\ref{assum:fullRank} reflects any deviations from the ideal deterministic LTI setting and may also arise from measurement noise, nonlinearities, time variance, etc.  For more details on Assumption~\ref{assum:fullRank} and discussions in the DPC setting, see \cite{Klaedtke2025_AT}. To establish the need for regularization $H(\Gb, \Pb) \neq 0$ in \eqref{eq:lqrOptDD}, we first consider the effects of choosing $H(\Gb, \Pb) = 0$.
\begin{prop}\label{prop:fail}
    Under Assumption~\ref{assum:fullRank} and with ${H(\Gb, \Pb) = 0}$, the optimal gain obtained from \eqref{eq:lqrOptDD} is given by $\overline\Kb =  \zerob.$
\end{prop}
\begin{proof}
    Under Assumption~\ref{assum:fullRank}, there exists a $\Gb$ satisfying \eqref{eq:ddPredCL} for any l.h.s.; as $\Gb$ is absent from all other constraints or costs, we may disregard $\Gb$ and \eqref{eq:ddPredCL}. This gives 
    \begin{align*}
    \min_{\Pb \succeq \Ib, \Kb, \Acl} \! \trace{\Qb\Pb} + \trace{\Kb^\top \Rb \Kb \Pb} \; \text{s.t.} \; \Acl\Pb\Acl^\top - \Pb + \Ib \preceq \zerob, \nonumber
\end{align*}
    where $\Kb$ and $\Acl$ can be freely chosen. Clearly, the constraint is always feasible, e.g., with $\Acl=\zerob$. Furthermore, since $\Rb$ and $\Pb$ are positive definite, the term $\trace{\Kb^\top \Rb \Kb \Pb}$ is uniquely minimized by $\overline\Kb = \zerob$.
\end{proof}
An equivalent result is provided by \cite[Thm.~2]{Zeng2024} in a stochastic setting, where Assumption~\ref{assum:fullRank} holds with probability 1. Note that this effect does not occur when using the covariance parameterization \eqref{eq:ddPredCLcovar}.
\begin{lem}\label{lem:covarImplicitConstraint}
    If \eqref{eq:GPE} is satisfied, the covariance parametrization \eqref{eq:ddPredCLcovar} implicitly enforces $\Acl 
    = \Als + \Bls \Kb$, where
    \begin{equation}
\!\begin{pmatrix}
        \Als & \Bls
\end{pmatrix}:= \Mb_\text{LS} = \argmin_\Mb \froSq{\Xb_1\! - \Mb \Dbc_0} = \Xb_1 \Dbc_0^+\! \label{eq:AB_LS}
    \end{equation}
are least-squares estimates of $(\Ab, \Bb)$.
\end{lem}
\begin{proof}
    Part of \eqref{eq:ddPredCLcovar} involves the empirical state-input covariance matrix $\hat\Sigma_{\Dbc_0} := \ell^{-1}\Dbc_0\Dbc_0^\top$, which is non-singular when \eqref{eq:GPE} is satisfied. Thus, we have $\Vb = 
    \hat\Sigma_{\Dbc_0}^{-1}\begin{pmatrix}
        \Ib & \Kb^\top
    \end{pmatrix}^\top$, which implies
    $\Acl = \frac{1}{\ell}\Xb_1 \Dbc_0^\top \Vb
    = \Als + \Bls \Kb$.
\end{proof}
As a consequence of Lemma~\ref{lem:covarImplicitConstraint}, data-driven LQR with covariance parametrization \eqref{eq:ddPredCLcovar} and $H(\Vb, \Pb) = 0$ coincides with a certainty-equivalent LQR \cite{Dörfler2023_certaintyEqLQR}, which was also shown in \cite{Zhao2025_DeePO-LQR}. The benefits of different regularizations $H(\Gb, \Pb)$ or $H(\Vb, \Pb)$ regarding robustness for the noisy setting have been investigated in previous works \cite{DePersis2021, Dörfler2022_roleOfRegLQR, Zhao2025_regularization}. In this paper, we further enhance the theoretical understanding of these regularizations via the analysis tool proposed in the next section.

\subsection{Parametric effects of quadratic regularizations}\label{sec:parametricEffect}

In this section, we analyze several choices of quadratic regularization for data-driven LQR proposed in \cite{DePersis2021, Dörfler2023_certaintyEqLQR, Zhao2025_regularization}. The key tool is analogous to the trajectory-specific effect in DPC \cite[Def.~1]{Klaedtke2025_AT}. However, since the l.h.s. of the parameterization \eqref{eq:ddPredCL} is not a trajectory triple, but rather involves the closed-loop system parameters $(\Kb, \Acl)$, we adapt the definition as follows.
\begin{defn}\label{def:parametricEffect}
    For a given $(\Kb, \Acl, \Pb)$, we call 
\begin{align}
    &H^\ast(\Kb, \Acl, \Pb) \! :=  \min_{\Gb} 
  H(\Gb, \Pb) \;  \text{s.t.} 
 \;\;  \Gb \;\; \text{satisfies} \;\eqref{eq:ddPredCL} \label{eq:OP_parametricEffect}
\end{align} 
the \textit{parametric effect} of $H(\Gb, \Pb)$ given $\Dbc$.
\end{defn}

Conceptually, \eqref{eq:OP_parametricEffect} is a multiparametric optimization problem with parameters $(\Kb, \Acl, \Pb)$, which naturally emerges as an inner optimization problem to \eqref{eq:lqrOptDD}. It assigns to every triple $(\Kb, \Acl, \Pb)$ the optimal regularization cost $H^\ast(\Kb, \Acl, \Pb)$, allowing intuitive interpretations through actual system quantities. Furthermore, under Assumption~\ref{assum:fullRank}, the variable $\Gb$ and constraint \eqref{eq:ddPredCL} can be eliminated after replacing $H(\Gb, \Pb)$ with $H^\ast(\Kb, \Acl, \Pb)$, leading to a lower-dimensional parameterization that does not scale with the number $\ell$ of data columns, which is shown in Section~\ref{sec:compBenefits}. Next, we analyze the parametric effect of the quadratic regularization 
\begin{equation}
    H(\Gb, \Pb) = \lambda  \, \trace{\Gb \Pb \Gb^\top} =\lambda\froSq{\Gb\Pb^\halff } \label{eq:quadraticReg}
\end{equation} introduced in \cite{DePersis2021} to promote robustness. 
\begin{prop}\label{prop:weightedQuadReg}
    Under Assumption~\ref{assum:fullRank}, the parametric effect of quadratic regularization \eqref{eq:quadraticReg} is given by
     \begin{align}
        H^\ast(\Kb, \Acl, \Pb) 
        =& \frac{\lambda}{\ell}\froSq{{\hat\Sigma_{\Delta\Xb}^{-\halff}} (\Acl-(\Als + \Bls \Kb))\Pb^\halff} \nonumber \\
        & \quad + \frac{\lambda}{\ell}\froSq{{\hat\Sigma_{\Delta\Ub}^{-\halff}} (\Kb-\Kls)\Pb^\halff} \nonumber\\
        & \quad + \frac{\lambda}{\ell}\froSq{{\hat\Sigma_{\Xb_0}^{-\halff}}\Pb^\halff} \label{eq:weightedQuadReg}
    \end{align}
    with $(\Als, \Bls)$ as in \eqref{eq:AB_LS}, and a least-squares controller
    \begin{equation}
        \Kb_\text{LS} := \argmin_\Kb \froSq{\Ub_0 - \Kb \Xb_0} = \Ub_0 \Xb_0^+. \label{eq:K_LS}
    \end{equation}%
    Further, $\hat\Sigma_{\Delta\Xb} := \ell^{-1}\Delta\Xb\Delta\Xb^\top$  and  $\hat\Sigma_{\Delta\Ub} := \ell^{-1}\Delta\Ub\Delta\Ub^\top$
    are empirical covariances of the residuals
    $\Delta\Xb := \Xb_1 - \Mb_\text{LS} \Dbc_0$ and 
    $\Delta\Ub := \Ub_0 - \Kb_\text{LS} \Xb_0$, and
    $\hat\Sigma_{\Xb_0}:=\ell^{-1} \Xb_0\Xb_0^\top$ is the empirical covariance matrix of the initial state $\xb_0$.
\end{prop}
\begin{proof}
    The change of variables $\widetilde \Gb :=  \Gb \Pb^\halff$ turns \eqref{eq:ddPredCL} into
    $$
        \begin{pmatrix}
        \Ib \\
        \Kb\\
        \Acl
    \end{pmatrix} =
    \begin{pmatrix}
        \Xb_0 \\
        \Ub_0 \\        
        \Xb_1
    \end{pmatrix}\Gb \quad \iff \begin{pmatrix}
        \Pb^\halff \\
        \Kb \Pb^\halff\\
        \Acl \Pb^\halff
    \end{pmatrix} =
    \begin{pmatrix}
        \Xb_0 \\
        \Ub_0 \\        
        \Xb_1
    \end{pmatrix}\widetilde\Gb.
    $$
    Next, we consider the columns of the involved matrices, i.e., 
    $\widetilde\Gb = \begin{pmatrix}
        \tilde\gb_1 & \hdots & \tilde\gb_n
    \end{pmatrix}$, 
    $ \Pb^\halff = \begin{pmatrix}
        \pb_1 & \hdots & \pb_n
    \end{pmatrix}$, 
    $\Kb \Pb^\halff = \begin{pmatrix}
        \Kb \pb_1 & \hdots & \Kb \pb_n
    \end{pmatrix}$, and
     $\Acl  \Pb^\halff = \begin{pmatrix}
        \Acl  \pb_1 & \hdots & \Acl  \pb_n
    \end{pmatrix}$
    to equivalently rewrite the optimization problem \eqref{eq:OP_parametricEffect} as
    \begin{align*}
        H^\ast(\Kb, \Acl, \Pb) =\min_{\tilde\gb_1, \hdots, \tilde\gb_n} & \sum_{i=1}^n \|\tilde\gb_i\|_2^2 \quad \mathrm{s.t.} \\
        \begin{pmatrix}
            \pb_i \\
            \Kb\pb_i \\
            \Acl\pb_i
        \end{pmatrix}&=\begin{pmatrix}
        \Xb_0 \\
        \Ub_0 \\        
        \Xb_1
    \end{pmatrix} \tilde\gb_i \quad \text{for}\; i \in \{1, \hdots, n\}.
    \end{align*}
    As columns $i\neq j$ do not interact in either the objective or constraints, the optimization decouples into $n$ problems 
    \begin{align*}
        H^\ast(\Kb, \Acl, \Pb) =\sum_{i=1}^n \min_{\tilde\gb_i}   \|\tilde\gb_i\|_2^2 \quad \mathrm{s.t.} \; \begin{pmatrix}
            \pb_i \\
            \Kb\pb_i \\
            \Acl\pb_i
        \end{pmatrix}=\begin{pmatrix}
        \Xb_0 \\
        \Ub_0 \\        
        \Xb_1
    \end{pmatrix} \tilde\gb_i.
    \end{align*}
    Each problem structurally matches the \textit{trajectory-specific effect of regularization} in explainable DPC \cite[Def.~1]{Klaedtke2025_AT}. Specifically, we may use the trajectory-specific effect $h^\ast(\xb_0, \ub, \xb)$ \cite[Eq.~(10)]{Klaedtke2025_AT} of quadratic regularization $h(\ab) = \lambda \|\ab\|_2^2$ for a state-space setting with prediction horizon $N=1$.
    Evaluating $h^\ast(\pb_i,
            \Kb\pb_i,
            \Acl\pb_i)$ yields
    \begin{align*}
        H^\ast(\Kb, \Acl, \Pb) =& \frac{\lambda}{\ell}\sum_{i=1}^n  \|\Acl\pb_i-\hat\xb_\mathrm{LS}(\pb_i, \Kb\pb_i)\|_{\hat\Sigma_{\Delta\Xb}^{-1}}^2 \\
        & \qquad\quad+ \|\Kb\pb_i - \hat\ub_\mathrm{LS}(\pb_i) \|_{\hat\Sigma_{\Delta\Ub}^{-1}}^2 + \|\pb_i\|_{\hat\Sigma_{\Xb_0}^{-1}}^2, 
    \end{align*}
    where $\hat\xb_\text{LS}(\xb_0, \ub) := \Als \xb_0 + \Bls \ub$  and 
     $\hat\ub_\text{LS}(\xb_0) := \Kb_\text{LS} \xb_0$.
    Further evaluating $\hat\xb_\mathrm{LS}(\pb_i, \Kb\pb_i)$ and $\hat\ub_\mathrm{LS}(\pb_i)$ leads to
    \begin{align*}
            H^\ast(\Kb, \Acl, \Pb) 
            =&\frac{\lambda}{\ell}\sum_{i=1}^n  \|\Acl\pb_i-\left(\Als\pb_i+\Bls\Kb\pb_i\right)\|_{\hat\Sigma_{\Delta\Xb}^{-1}}^2 \\
            &\qquad \quad  + \|\Kb\pb_i - \Kls \pb_i \|_{\hat\Sigma_{\Delta\Ub}^{-1}}^2 + \|\pb_i\|_{{\hat\Sigma_{\Xb_0}^{-1}}}^2,
    \end{align*}
    and translating the sum of squared 2-norms for individual columns back to a squared Frobenius norm yields \eqref{eq:weightedQuadReg}. 
\end{proof}
We can see that the synthesized $\Acl$ is pushed towards a closed-loop matrix $\Als + \Bls \Kb$ involving the least-squares estimates $\Als, \Bls$ and the synthesized gain $\Kb$. This push occurs more (less) harshly in directions in which the least-squares predictor $\hat\xb_\text{LS}(\xb_0, \ub) = \Als \xb_0+\Bls\ub$ is considered to be more (less) accurate based on the empirical covariance $\hat\Sigma_{\Delta\Xb}$. The same effect occurs for the synthesized gain $\Kb$ being pushed towards a least-squares estimate $\Kls$, again, weighted with the (inverse of the) corresponding covariance $\hat\Sigma_{\Delta\Ub}$. Note that $\Kls$ represents the best-explored controller rather than being optimized for performance. The deviation between the synthesized pair $(\Kb, \Acl)$ and their least-squares counterparts $(\Kls, \Als + \Bls \Kb)$ is not only weighted with the respective inverse covariances ${\hat\Sigma_{\Delta\Xb}^{-1}}, {\hat\Sigma_{\Delta\Ub}^{-1}}$ but also with the closed-loop controllability Gramian $\Pb$, which remains an optimization variable in \eqref{eq:lqrOptDD}. Therefore, the third cost term also includes an optimization variable of \eqref{eq:lqrOptDD}, marking a notable distinction from the trajectory-specific effect in DPC \cite[Eq.~(10)]{Klaedtke2025_AT}. Its effect can be made apparent by realizing that part of the objective function in \eqref{eq:lqrOptDD} reads
$
    \trace{\Qb\Pb} + \nicefrac{\lambda}{\ell}\,\trace{{\hat\Sigma_{\Xb_0}^{-1}}\Pb} = \trace{\widetilde\Qb \Pb}$ with $\widetilde\Qb := \Qb + \nicefrac{\lambda}{\ell}\,{\hat\Sigma_{\Xb_0}^{-1}}. 
$
Hence, for increasing $\nicefrac{\lambda}{\ell}$ this new weight matrix $\widetilde\Qb$ becomes increasingly dominated by $\hat\Sigma_{\Xb_0}^{-1}$, which penalizes well (poorly) explored directions of the state-space more (less) harshly. Finally, since $H(\Gb, \Pb) = \lambda\froSq{\Gb\Pb^\halff}$ and $H(\Gb, \Pb) = \lambda\fro{\Gb\Pb^\halff}$ share optimizers by monotonicity, the parametric effect of $H(\Gb, \Pb) = \lambda\fro{\Gb\Pb^\halff}$ can be obtained based on Proposition~\ref{prop:weightedQuadReg} by taking the square root. This applies to all the regularizations considered here. 

The next result considers the projection-based quadratic regularization $H(\Gb, \Pb) = \lambda\froSq{\Pib_\perp \Gb \Pb^\halff}$ with the projector $\Pib_\perp := \Ib - \Dbc_0^+ \Dbc_0$, which is related to the more general projection-based regularization $H(\Gb, \Pb) = \lambda\|\Pib_\perp \Gb\|$ (for any norm $\|.\|$) proposed in \cite{Dörfler2023_certaintyEqLQR} for enforcing a certainty equivalent solution.
\begin{prop}\label{prop:projectionBasedWeightedQuadReg}
    Under Assumption~\ref{assum:fullRank}, the parametric effect of the projection-based quadratic regularization $H(\Gb,\Pb) = \lambda\froSq{\Pib_\perp\Gb\Pb^\halff} = \lambda\,\trace{\Pib_\perp\Gb \Pb \Gb^\top}$ is given by
    \begin{equation*}
        H^\ast(\Kb, \Acl, \Pb) \!
        = \frac{\lambda}{\ell}\froSq{{\hat\Sigma_{\Delta\Xb}^{-\halff}} (\Acl \! - \!(\Als + \Bls \Kb))\Pb^\halff}.  
    \end{equation*}
\end{prop}
\begin{proof}
    The proof is analogous to Proposition~\ref{prop:weightedQuadReg}
        by using
        $\Pib_\perp \Gb \Pb^\halff = 
        \begin{pmatrix}
            \Pib_\perp \tilde\gb_1 & \hdots & \Pib_\perp \tilde\gb_n
        \end{pmatrix}$,
    and the trajectory-specific effect of $h(\ab) = \|\Pib_\perp\ab\|_2^2$ from \cite[Eq.~(15)]{Klaedtke2025_AT}. 
\end{proof}
The interpretations of this cost term are equivalent to those discussed below Proposition~\ref{prop:weightedQuadReg}. Furthermore, for $\nicefrac{\lambda}{\ell} \to \infty$, the solution of \eqref{eq:lqrOptDD} tends towards that of a certainty-equivalent formulation \cite{Dörfler2023_certaintyEqLQR}, where $\Acl = \Als + \Bls \Kb$ is enforced as a hard constraint.
\begin{rem}
        While not explicitly included in the notation of $\|\Pib_\perp \Gb\|$, the convex reformulation considered in \cite[Sec.~III.C]{Dörfler2023_certaintyEqLQR} also makes use of $\Pb$. More specifically, the orthogonality constraint $\Pib_\perp\Gb = \zerob$ is first equivalently formulated as $\Pib_\perp\Gb\Pb = \zerob$ before being lifted to the objective function. Therefore, the associated regularization is $H(\Gb, \Pb) = \lambda\|\Pib_\perp \Gb \Pb\|$ instead of $H(\Gb) = \lambda\|\Pib_\perp \Gb\|$.
    \end{rem}

We next characterize the quadratic regularization 
\begin{equation}
    H(\Vb, \Pb) = \trace{\Vb \Pb \Vb^\top \hat\Sigma_{\Dbc_0}} = \lambda \froSq{\hat\Sigma_{\Dbc_0}^\halff \Vb \Pb^\halff} \label{eq:covarReg}
\end{equation} proposed in \cite{Zhao2025_regularization} for the covariance parametrization \eqref{eq:ddPredCLcovar}.
\begin{prop} \label{prop:covarReg}
    Consider the covariance parameterization \eqref{eq:ddPredCLcovar}     replacing \eqref{eq:ddPredCL} in Definition~\ref{def:parametricEffect}, and assume $\Dbc_0$ satisfies \eqref{eq:GPE}. Then,     the parametric effect of quadratic regularization \eqref{eq:covarReg} is given by
    \begin{align*}
        H^\ast(\Kb, \Pb) 
        =& \lambda \froSq{\hat\Sigma_{\Delta\Ub}^{-\halff} (\Kb-\Kls)\Pb^\halff}+ \lambda \froSq{\hat\Sigma_{\Xb_0}^{-\halff}\Pb^\halff}.
    \end{align*}
\end{prop}
\begin{proof}
    Since $\Vb = \hat\Sigma_{\Dbc_0}^{-1}\begin{pmatrix}
        \Ib & \Kb^\top
    \end{pmatrix}^\top$ is fixed by a given $\Kb$, solving \eqref{eq:OP_parametricEffect} with the covariance parameterization \eqref{eq:ddPredCLcovar} replacing \eqref{eq:ddPredCL} involves no optimization. Instead, we evaluate
    $$
        \trace{\Vb\Pb\Vb^\top\hat\Sigma_{\Dbc_0}} 
        = \mathrm{tr}\left(\hat\Sigma_{\Dbc_0}^{-1}\begin{pmatrix}
            \Ib \\ \Kb
        \end{pmatrix}
        \Pb\begin{pmatrix}
            \Ib \\ \Kb
        \end{pmatrix}^\top\right),
    $$
    where the r.h.s. reformulation was also done in \cite[Eq.~(28)]{Zhao2025_regularization}. 
    One may verify that the involved inverse admits the block form
    $$
    \hat\Sigma_{\Dbc_0}^{-1} = 
        \begin{pmatrix}
            \Ib & \zerob \\
            -\Kls & \Ib
        \end{pmatrix}^\top \begin{pmatrix}
            \hat\Sigma_{\Xb_0}^{-1} & \zerob \\
            \zerob & \hat\Sigma_{\Delta\Ub}^{-1}
        \end{pmatrix}\begin{pmatrix}
            \Ib & \zerob \\
            -\Kls & \Ib
        \end{pmatrix},
    $$
    and we use this expression to compute 
    $$
        \begin{pmatrix}
        \Ib \\ \Kb
    \end{pmatrix}^\top \hat\Sigma_{\Dbc_0}^{-1} \begin{pmatrix}
        \Ib \\ \Kb
    \end{pmatrix} = \hat\Sigma_{\Xb_0}^{-1} + (\Kb - \Kls)^\top \hat\Sigma_{\Delta \ub}^{-1} (\Kb - \Kls).
    $$
    Thus, cyclically rotating the trace yields
    \begin{align*}
        H^\ast(\Kb, \Pb) &= \lambda\, \mathrm{tr}\left(\begin{pmatrix}
        \Ib \\ \Kb
    \end{pmatrix}^\top\hat\Sigma_{\Dbc_0}^{-1}\begin{pmatrix}
        \Ib \\ \Kb
    \end{pmatrix}
    \Pb\right) \\
    &= \lambda\, \trace{\hat\Sigma_{\Xb_0}^{-1} \Pb + (\Kb - \Kls)^\top \hat\Sigma_{\Delta \ub}^{-1} (\Kb - \Kls)\Pb} \\
    &= \lambda \froSq{\hat\Sigma_{\Delta\Ub}^{-\halff} (\Kb-\Kls)\Pb^\halff}+ \lambda \froSq{\hat\Sigma_{\Xb_0}^{-\halff}\Pb^\halff}, 
    \end{align*} 
    which completes the proof. 
\end{proof}
This expression mirrors \eqref{eq:weightedQuadReg} with two distinctions: (i) the first cost term of \eqref{eq:weightedQuadReg} is absent since $\Acl=\Als+\Bls\Kb$ is enforced by the covariance parameterization, and (ii) the terms are not scaled by $\ell^{-1}$ due to the scaling in \eqref{eq:ddPredCLcovar}. This avoids performance degradation as $\ell$ increases \cite[Thm.3]{Zeng2024}. A similar effect occurs in DPC \cite{Mattsson2024} and can be mitigated by scaling $\lambda$ with $\ell$ or embedding this scaling in the parameterization, as in \eqref{eq:ddPredCLcovar}. 
\begin{rem}
    Proposition~\ref{prop:weightedQuadReg}, \ref{prop:projectionBasedWeightedQuadReg}, and~\ref{prop:covarReg} are also valid for $\Pb$ being replaced by an arbitrary symmetric positive definite weight matrix $\Sb\in\R^{n\times n}$.  However, $\Sb = \Pb$ is usually considered, since it facilitates convex reformulations of \eqref{eq:lqrOptDD}.
\end{rem}

\subsection{Computational benefits and mixed regularizations} \label{sec:compBenefits}

As mentioned below Definition~\ref{def:parametricEffect}, the parametric effect $H^\ast(\Kb, \Acl, \Pb)$ can be used in place of $H(\Gb, \Pb)$. Further, under Assumption~\ref{assum:fullRank}, the constraint \eqref{eq:ddPredCL} and $\Gb$ can be fully eliminated. The resulting optimization problem reads as
\begin{align}
    \min_{\Pb \succeq \Ib, \Kb, \Acl} &\trace{\Qb\Pb} + \trace{\Kb^\top \Rb \Kb \Pb} + H^\ast(\Kb, \Acl, \Pb) \label{eq:lqrOptParamEffect}\\
    &\text{s.t.} \quad \Acl\Pb\Acl^\top - \Pb + \Ib \preceq \zerob, \nonumber
\end{align}
and we briefly demonstrate a convex reformulation for the parametric effect of $H(\Gb, \Pb) = \lambda\,\trace{\Gb\Pb\Gb^\top}$. Using \eqref{eq:weightedQuadReg} and the change of variables $\tKb := \Kb \Pb$ and $\tAcl:=\Acl \Pb$,  a convex reformulation of \eqref{eq:lqrOptParamEffect} is given by
\footnotesize
\begin{align}
    &  \min_{\Pb, \tilde{\Kb}, \tilde{\Ab}_\text{cl}, \Lb, \Mb, \Nb} \trace{\Qb\Pb} + \trace{\Rb \Lb} \label{eq:convexReformulation}\\& \qquad \qquad \qquad  +\frac{\lambda_1}{\ell}\trace{\Mb \hat\Sigma_{\Delta\Xb}} + \frac{\lambda_2}{\ell}\trace{\Nb \hat\Sigma_{\Delta\Ub}} + \frac{\lambda_3}{\ell}\trace{\hat\Sigma_{\Xb_0}\Pb} \nonumber \\
    & \quad\text{s.t.} \begin{pmatrix}
        \Pb-\Ib & \tAcl \\
        \tAcl^\top & \Pb
    \end{pmatrix} \succeq \zerob,\quad \Pb \succeq \Ib,  \quad  \nonumber\\
     & \qquad\,\begin{pmatrix}
        \Lb & \tKb \\
        \tKb^\top & \Pb
    \end{pmatrix} \succeq \zerob, \quad\begin{pmatrix}
        \Nb & \tKb - \Kls\Pb \\
        \left(\tKb - \Kls\Pb\right)^\top & \Pb
    \end{pmatrix} \succeq \zerob, \nonumber\\ & \qquad\,\begin{pmatrix}
        \Mb & \tAcl - \left(\Als \Pb + \Bls \tKb\right) \\
        \left(\tAcl-\left(\Als \Pb + \Bls \tKb\right)\right)^\top \!\!\!\!\!& \Pb
    \end{pmatrix} \succeq \zerob \nonumber 
\end{align}
\normalsize
with $\lambda = \lambda_1 = \lambda_2 = \lambda_3$, and symmetric positive definite $\Lb\in \R^{m\times m}, \Mb\in \R^{n\times n}, \Nb\in \R^{m\times m}$. The reformulation is stated without proof, as it follows standard procedures also used in \cite{DePersis2021, Dörfler2023_certaintyEqLQR, Zhao2025_regularization}. While introducing additional optimization variables compared to \cite{DePersis2021}, the size of \eqref{eq:convexReformulation} does not scale with the number of data columns $\ell$. Furthermore, seeing as the parametric effect \eqref{eq:weightedQuadReg} consists of three parts with drastically different interpretations, the reformulation may be used to weight them individually by choosing separate non-negative $\lambda_1, \lambda_2, \lambda_3$ for each term. This idea also appears in mixed regularization experiments \cite[Sec~V.B]{Dörfler2022_roleOfRegLQR} and in DPC \cite{Breschi2022new, Klaedtke2025_AT}. In the same vein, a convex reformulation for the projection-based variation in Proposition~\ref{prop:projectionBasedWeightedQuadReg} can be obtained from \eqref{eq:convexReformulation} by choosing $\lambda = \lambda_1$, $\lambda_2 = \lambda_3 = 0$, and dropping the variable $\Nb$ along with its constraint. Finally, the same ideas can be applied to the covariance parameterization and its regularization, resulting in the convex reformulation
\footnotesize
\begin{align}
    &  \min_{\Pb, \tilde\Kb, \Lb, \Nb} \trace{\Qb\Pb} + \trace{\Rb \Lb}  + \lambda_2\trace{\Nb \hat\Sigma_{\Delta\Ub}} + \lambda_3\trace{\hat\Sigma_{\Xb_0}\Pb} \label{eq:convexReformulationCovar}\\
    & \quad\text{s.t.} \begin{pmatrix}
        \Pb-\Ib & \Als \Pb + \Bls \tKb \\
        \left(\Als \Pb + \Bls \tKb\right)^\top & \Pb
    \end{pmatrix} \succeq \zerob,\quad \Pb \succeq \Ib,  \quad  \nonumber\\
     & \qquad\,\begin{pmatrix}
        \Lb & \tKb \\
        \tKb^\top & \Pb
    \end{pmatrix} \succeq \zerob, \quad\begin{pmatrix}
        \Nb & \tKb - \Kls\Pb \\
        \left(\tKb - \Kls\Pb\right)^\top & \Pb
    \end{pmatrix} \succeq \zerob \nonumber
\end{align}
\normalsize
with $\tKb = \Kb \Pb$, $\lambda = \lambda_2 = \lambda_3$, and symmetric positive definite $\Lb\in \R^{m\times m}, \Nb\in \R^{m\times m}$. While the computational benefit is limited, since the covariance parameterization already does not scale with $\ell$, this reformulation may also be used to choose $\lambda_2, \lambda_3$ individually, if desired.

\section{Numerical experiment}\label{sec:experiments}

As we have shown, the parametric effects of the considered regularizations all involve (subsets of) the terms
\begin{align*}
    H_1^\ast(\Kb, \Acl, \Pb)  &=  \froSq{{\hat\Sigma_{\Delta\Xb}^{-\halff}} (\Acl-(\Als + \Bls \Kb))\Pb^\halff}, \\
     H_2^\ast(\Kb, \Pb) &=\froSq{{\hat\Sigma_{\Delta\Ub}^{-\halff}} (\Kb-\Kls)\Pb^\halff}, \\
        H_3^\ast( \Pb) &=  \froSq{{\hat\Sigma_{\Xb_0}^{-\halff}}\Pb^\halff}. 
\end{align*}
For notational convenience, we will drop the argument when referring to individual $H_i^\ast$. Also, we assume that regularization weights have been appropriately scaled with $\ell$ to allow comparison between the two parameterizations \eqref{eq:ddPredCL} and \eqref{eq:ddPredCLcovar}. In this section, we conduct a numerical experiment specifically designed to isolate (combinations of) terms  $H_i^\ast$, allowing us to examine their influence on the resulting data-driven LQR solution. For easier notation, we associate a set $\Sc$ with $H^\ast = \sum_{i\in \Sc} \lambda H_i^\ast$. For instance, the cases $\Sc = \{1, 2, 3\}$ and $\Sc = \{1\}$ correspond to the parametric effects in Proposition~\ref{prop:weightedQuadReg} and~\ref{prop:projectionBasedWeightedQuadReg}, respectively. These different cases $\Sc$ are implemented via \eqref{eq:convexReformulation} or \eqref{eq:convexReformulationCovar} by setting $\lambda_i = \lambda$ for $i\in \Sc$ and $\lambda_j = 0$ for $j\notin \Sc$ (and removing irrelevant variables and constraints). For visualization purposes, we consider a low-dimensional LTI system \eqref{eq:LTI} with
$$
    \Ab = \begin{pmatrix}
        0.525 & -0.325 \\
        -0.325 & 0.525
    \end{pmatrix}, \quad \Bb = \begin{pmatrix}
        1 \\ 0
    \end{pmatrix},
$$
and i.i.d. disturbance $\wb(k)$ drawn from $\Nc(\zerob, 0.01\cdot \Ib)$. Furthermore, we choose $\Qb = \Ib$, $\Rb = 0.1$, and $\ell = 30$. To visualize the directional aspects of $H_2^\ast$ and $H_3^\ast$, we explore the input-state space in a very specific manner. We choose $\Xb_0 = 10 \cdot \oneb^\top \vb + \Xb_\text{rnd}$ and $\Ub_0 = \Kb_\text{expl} \Xb_0 + \Ub_\text{rnd}$ with $\vb = \begin{pmatrix}
    -1 & 1
\end{pmatrix}^\top$ and the (barely stabilizing) exploration controller $\Kb_\text{expl}:= \begin{pmatrix}
    -2.8 & \,6.8
\end{pmatrix}$. The columns of the excitation components $\Xb_\text{rnd}$ and $\Ub_\text{rnd}$ are drawn i.i.d. from $\Nc(0, \Ib)$ and $\Nc(0, 1)$, respectively. Due to their comparatively low variance, exploration is good near $\vb$ and near $\Kb_\text{expl}$, but poor in other regions. This will be reflected by $\Kls \approx \Kb_\text{expl}$ and by the eigenpairs of $\hat\Sigma_{\Xb_0}$ and $\hat\Sigma_{\Delta\Ub}$. Due to space restrictions, we only demonstrate an experiment outcome in which $\Kls$ stabilizes $(\Als, \Bls)$. SDPs are solved using \cite{cvx} with SDPT3 4.0 \cite{SDPT3}. The code used to generate these figures is available online\footnote{\scriptsize{\href{https://www.github.com/Control-and-Cyberphysical-Systems/ddLQR-quadReg}{github.com/Control-and-Cyberphysical-Systems/ddLQR-quadReg}}}. 

First, we focus on the term $H_1^\ast$, which pushes $\Acl$ towards $\Als + \Bls\Kb$. To see how it is traded off with the other terms, we examine all combinations of the terms $H_i^\ast$ that involve $H_1^\ast$. That is, we implement the cases $\{1, 2, 3\}, \{1\}, \{1, 2\}$, and $\{1, 3\}$ via \eqref{eq:convexReformulation}. For comparison with \{1, 2, 3\} and \{1\}, we further implement $H(\Gb, \Pb) = \lambda\,\trace{ \Gb \Pb \Gb^\top}$ and $H(\Gb, \Pb) = \lambda\,\trace{\Pib_\perp \Gb \Pb \Gb^\top}$.  Figure~\ref{fig:Experiment1} shows the deviation of the synthesized closed-loop matrix $\fro{\overline{\Ab}_\text{cl}-(\Als+\Bls\overline{\Kb})}$ for each case over a range of $\lambda\in[10^{-4}, 10^{6}]$. The results indeed confirm Propositions~\ref{prop:weightedQuadReg} and~\ref{prop:projectionBasedWeightedQuadReg}; the discrepancies (blue) observed in Figure~\ref{fig:Experiment1} coincide with numerical issues reported by the solver. Even though the controllers $\overline\Kb$ converge to very different values, we can see similar behavior of $\overline{\Ab}_\text{cl}$ for the pairs $\{\{1\}, \{1, 2\}\}$ and $ \{\{1,3\}, \{1, 2, 3\}\}$, respectively. In the former, $\overline{\Ab}_\text{cl}$ indeed tends (non-monotonically) towards $\Als+\Bls\overline{\Kb}$ for increasing $\lambda$. In the latter, the presence of $H_3^\ast$ seems to prevent this behavior. 
\begin{figure}
    \centering
    \includegraphics[trim=4cm 11.91cm 5cm 11.99cm,clip=true, width=\linewidth]{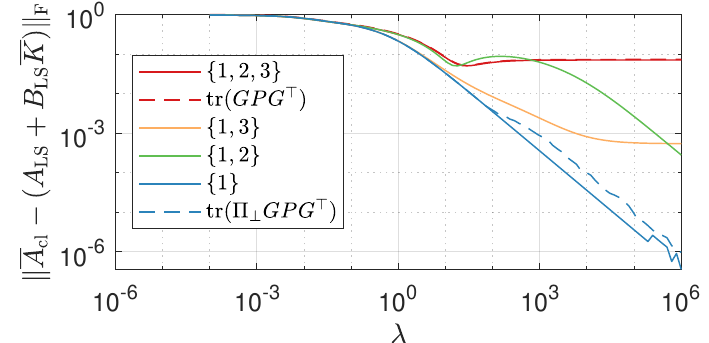}\vspace{-3mm}
    \caption{Effect of $H_1^\ast$ for different regularizations shown via the deviation between $\overline{\Ab}_\text{cl}$ and $\Als + \Bls \overline\Kb$ for $\lambda \in[10^{-6},10^{6} ]$.} 
    \label{fig:Experiment1}
\end{figure}

Next, we study $H_2^\ast$, which pushes $\Kb$ towards $\Kls$. Recall that $\Kls$ is a least-squares estimate based on $\Dbc_0$ (see \eqref{eq:K_LS}), and reflects the best-explored controller. We use the covariance parameterization \eqref{eq:ddPredCLcovar} since it ``removes'' $H_1^\ast$ by implicitly enforcing $\Acl = \Als + \Bls \Kb$, and implement the cases $\{2, 3\}$ and $\{2\}$ via \eqref{eq:convexReformulationCovar}. For comparison with \{2, 3\}, we also implement $H(\Vb, \Pb) = \trace{\Vb\Pb\Vb^\top \hat\Sigma_{\Dbc_0}}$. Figure~\ref{fig:Experiment2} shows the respective optimal gains $\overline{\Kb}$ parameterized by $\lambda\in[0, 10^{10}]$. The results indeed confirm Proposition~\ref{prop:covarReg}. All cases start at the certainty equivalent LQR gain $\Kb_\text{CE}$ \cite{Dörfler2023_certaintyEqLQR} for $\lambda = 0$, and move towards the best-explored controller $\Kls$ for increasing $\lambda$. However, $H_3^\ast$ seems to prevent the case $\{2, 3\}$ from fully converging to $\Kls$, which only occurs in $\{2\}$.
\begin{figure}
    \centering
    \includegraphics[trim=4cm 11.85cm 5cm 12.2cm,clip=true, width=\linewidth]{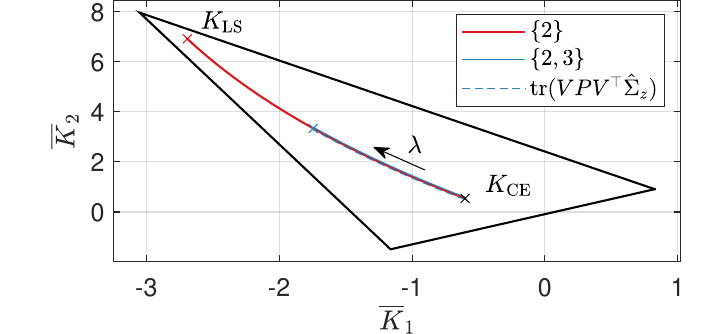}\vspace{-3mm}
    \caption{Effect of $H_2^\ast$ for different regularizations shown via the optimal gains $\overline{\Kb}$ parameterized by $\lambda \in [0, 10^{10}]$. The set of stabilizing $\Kb$ for the estimated $(\Als, \Bls)$ is shown in black. } 
    \label{fig:Experiment2}
\end{figure}

Next, we focus on the term $H_3^\ast$, which penalizes poorly explored regions of the state-space. We again ``remove'' the effect of $H_1^\ast$ by considering the covariance parameterization \eqref{eq:ddPredCLcovar}. Further, we fully isolate $H_3^\ast$ by only considering the case \{3\}, implemented via \eqref{eq:convexReformulationCovar}. Figure~\ref{fig:Experiment3} shows phase portraits and example trajectories of the synthesized closed-loop matrix $\overline{\Ab}_\text{cl}$. On their way to the origin, states are increasingly driven toward the best explored direction $\vb = \begin{pmatrix}
    -1 & 1
\end{pmatrix}^\top$ as $\lambda$ increases.
\begin{figure}
    \centering
    \includegraphics[trim=6cm 10.9cm 6cm 10.5cm,clip=true, width=\linewidth]{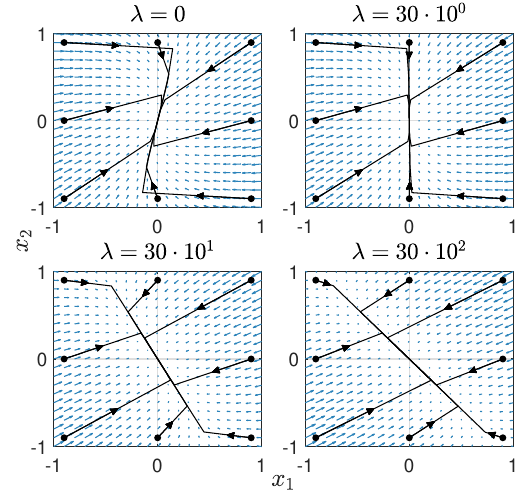} \vspace{-8mm}
    \caption{Effect of $H_3^\ast$ for the case $\{3\}$ shown via phase portraits and (disturbance-free) example trajectories of $\xb(k+1) = \overline{\Ab}_\text{cl}\xb(k)$ .}
    \vspace{-2mm}
    \label{fig:Experiment3}
\end{figure}
Finally, Tab.~\ref{tab:compTimes} reports computation times for various $\ell$, including the preprocessing required for the matrices in \eqref{eq:convexReformulation}, illustrating the computational advantages of \eqref{eq:convexReformulation}.
\begin{table}[]
\centering
\caption{Mean computation times over 10 runs on an Intel Core i7.}
\begin{tabular}{crrrr}\toprule
                                    & $\ell = 30$ & $\ell = 60$ & $\ell = 90$ & $\ell = 120$ \\\midrule
$\trace{\Gb\Pb\Gb}$            & \SI{1.31}{\s}        & \SI{2.60}{\s}        & \SI{11.97}{\s}        & \SI{44.26}{\s}        \\
$\{1, 2, 3\}$                       & \SI{0.51}{\s}        & \SI{0.51}{\s}        & \SI{0.47}{\s}        & \SI{0.49}{\s}         \\\midrule
$\trace{\Pib_\perp \Gb\Pb\Gb^\top}$ & \SI{1.66}{\s}        & \SI{4.57}{\s}        & \SI{22.59}{\s}        & \SI{62.28}{\s}        \\
$\{1\}$                             & \SI{0.52}{\s}        & \SI{0.50}{\s}        & \SI{0.46}{\s}         & \SI{0.48}{\s}        \\\bottomrule
\end{tabular}\label{tab:compTimes}
\end{table}

\section{Summary and Outlook}\label{sec:Conclusion}

We introduced the concept of the \textit{parametric effect of regularization} to explain how quadratic regularization influences direct data-driven LQR solutions. This concept translates regularization costs from auxiliary variables to system quantities, thereby enhancing interpretability and providing clearer insights into regularization mechanisms. Our analysis showed that different quadratic regularizations systematically push the synthesized closed-loop parameters toward least-squares estimates, weighted by empirical covariances of the involved data. Beyond improved explainability, we demonstrated that this idea enables the elimination of an auxiliary variable, removing the dependence of problem size on the number of data samples. Numerical experiments confirmed the distinct roles of individual parametric terms in shaping  synthesized closed-loop dynamics and controller gains. Future work may further investigate mixed regularization weighting strategies and their connection to established stability and performance bounds for the un-mixed case. Moreover, exploring further synergies with explainability results from the closely related DPC framework appears promising for advancing data-driven LQR design.

\end{document}